%% file: paper.tex
\documentclass{llncs}

\usepackage{amsmath}
\usepackage{amsfonts}

\newcommand{\ACC}{\mathsf{ACC}}
\newcommand{\BS}{\mathsf{BS}}
\newcommand{\Ctx}{\ensuremath{\mathsf{Ctx}}}
\newcommand{\FOL}{\ensuremath{\mathsf{FOL}}}
\newcommand{\KB}{\mathsf{KB}}
\newcommand{\app}{\mathsf{app}}
\newcommand{\mng}{\mathsf{mng}}
\newcommand{\br}{\mathsf{br}}
\newcommand{\kb}{\mathsf{kb}}
\newcommand{\pnot}{\mathsf{not}\ }
\newcommand{\U}{\mathcal U}
\newcommand{\myparagraph}[1]{\medskip\noindent\emph{#1}}

\begin{document}

\title{Integrity Constraints for General-Purpose Knowledge Bases}
\author{Lu\'\i s Cruz-Filipe\inst{1}
  \and Isabel Nunes\inst{2}
  \and Peter Schneider-Kamp\inst{1}}
\institute{
  Dept.\ of Mathematics and Computer Science,
  University of Southern Denmark
  \and Faculdade de Ci\^encias da Universidade de Lisboa, Portugal}

\maketitle

\begin{abstract}
  Integrity constraints in databases have been studied extensively
  since the 1980s, and they are considered essential to guarantee
  database integrity.
  In recent years, several authors have studied how the same notion
  can be adapted to reasoning frameworks, in such a way that they
  achieve the purpose of guaranteeing a system's consistency, but are
  kept separate from the reasoning mechanisms.

  In this paper we focus on multi-context systems, a general-purpose
  framework for combining heterogeneous reasoning systems,
  enhancing them with a notion of integrity constraints that
  generalizes the corresponding concept in the database world.
\end{abstract}

\input{introduction}
\input{background}
\input{ic}

\input{apps}
\input{repair}
\input{conclusion}

\bibliographystyle{plain}
\bibliography{bibl}

\end{document}

%% file: introduction.tex
\section{Introduction}
\label{sec:intro}

Integrity constraints in databases have now been around for decades, and are universally
acknowledged as one of the essential tools to ensure database
consistency~\cite{Abiteboul1995}.
The associated problem of finding out how to repair an inconsistent database -- i.e.,
change it so that it again satisfies the integrity constraints -- was soon recognized as
an important and difficult one~\cite{Abiteboul1988}, which would unlikely be solvable in a
completely automatic way~\cite{Eiter1992}.

Since the turn of the century, much focus in research has moved from classical
databases to more powerful reasoning systems, where information is not all explicitly
described, but may be inferred by logical means.
In this setting, an important topic of study is how to combine the reasoning
capabilities of different systems, preferrably preserving the properties that make them
useful in practice -- e.g.~consistency, decidability of reasoning, efficient computation.
One of the most general frameworks to combine reasoning systems abstractly is that of
heterogeneous nonmonotonic multi-context systems~\cite{Brewka2007b}.
Besides being studied from a theoretical perspective, these have been implemented, and
many specialized versions have been introduced to deal with particular aspects deemed
relevant in practice~\cite{Brewka2011,DaoTran2010,Goncalves2014,Tasharrofi2014}.
In this work, we will work with relational multi-context systems~\cite{Fink2011}, a
first-order generalization of the original, propositional-based systems, which we
will refer to simply as multi-context systems, or MCSs.

As a very simple kind of reasoning system, databases can naturally be viewed as particular
cases of MCSs.
In this paper we propose to define integrity constraints in MCSs in a way
that naturally generalizes the usual definitions for relational databases.
Some authors have previously discussed modelling integrity constraints in MCSs, but their
approach differs substantially from the typical database perspective, as integrity
constraints are embedded \emph{into} the system, thereby becoming part of the reasoning
mechanism -- unlike the situation in databases, where they form an independent layer that
simply signals whether the database is in a consistent state.
We argue that integrity constraints for MCSs should also follow this principle, and show
how our approach is also in line with investigations on how to add integrity constraints
to other reasoning frameworks, namely description logic knowledge
bases~\cite{Fang2011,Motik2009}.
Due to the richer structure of MCSs, we can define two distinct notions of consistency
with respect to integrity constraints, which coincide for the case of databases.

We also address the problem of repairing an MCS that does not satisfy its integrity
constraints by moving to managed multi-context systems (mMCSs)~\cite{Brewka2011}, which
offer additional structure that helps defining the notion of repair.

\myparagraph{Contributions.}
Our main contribution is a uniform notion of integrity constraint over several
formalisms.
We define integrity constraints over an MCS, together with notions of
weak and strong satisfaction of these.
We show that the problem of deciding whether an MCS satisfies a set of integrity
constraints is polynomial-time reducible to the problem of deciding whether an MCS is
logically consistent (i.e., it has a model).
We show how our definition captures the traditional notion of integrity constraints over
relational databases, and how it naturally generalizes this concept to distributed
databases and deductive databases.
We also compare our definition with existing proposals for integrity constraints over
ontology languages.
Finally, we define repairs, and show how our definition again generalizes the traditional
concept in databases.

\myparagraph{Outline.}
In Section~\ref{sec:background} we introduce the framework of multi-context systems.
In Section~\ref{sec:ic} we define integrity constraints over MCSs, together with the
notions of weak and strong satisfaction.
We show how we can encode an MCS with integrity constraints as a different MCS, and obtain
decidability and complexity results for satisfaction of integrity constraints by reducing
to the problem of logical consistency.
In Section~\ref{sec:apps} we justify our
definition of
integrity constraint, by showing that it generalizes the usual concept in relational
databases, as well as other authors' proposals for ontology languages~\cite{Motik2009} and
peer-to-peer systems~\cite{Caroprese2007}.
We also show that it induces a natural concept of integrity constraint for distributed
databases, as well as providing a similar notion for deductive databases that is more
expressive than the usual one; and provide complexity results for these concrete cases.
In Section~\ref{sec:repair} we recall the notion of a database repair, and show how
repairs can be naturally defined in a simple extension of MCSs.
We conclude with an overview of our results and future directions in
Section~\ref{sec:conclusion}.


\subsection{Related work}

The topic of integrity constraints has been extensively studied in the literature.
In this section, we discuss the work that we feel to be more directly relevant to the
tasks we carry out in this paper.

Integrity constraints and updates -- ways of repairing inconsistent databases -- were
identified as a seminal problem in database theory almost thirty years
ago~\cite{Abiteboul1988}.
The case for viewing integrity constraints as a layer on top of the database, rather than
as a component of it, has been made since the 1980s.
The idea is that the data inconsistencies captured by integrity constraints need to be
resolved, but they should not interfere with the ability to continue using the database.
In this line, much work has been done e.g.~in query answering from inconsistent
databases~\cite{Chomicki1999,Chomicki2010}, by ensuring that the only answers generated
are those that hold in minimally repaired versions of the database.

The first authors to consider deductive databases~\cite{Asirelli1985,Gallaire1984} also
discussed this issue.
They identify three ways to look at deductive databases: by viewing the whole system as a
first-order theory; by viewing it as an extensional database together with integrity
constraints; and a mixed view, where some rules are considered part of the logic theory
represented by the database, and others as integrity constraints identifying preferred
models.
In~\cite{Asirelli1985}, it is argued that this third approach is the correct one, as it
cleanly separates rules that are meant to be used in logic inferencing from those that
only specify consistency requirements.

More recently, authors have considered adding integrity constraints to open-world systems
such as ontologies.
Although integrity constraints can be written in the syntax of terminological axioms,
the authors of~\cite{Motik2009} discuss why they should still be kept separate from the
logical theory.
Therefore, they separate the axioms in the T-Box (the deductive part of an ontology) into
two groups: reasoning rules, which are used to infer new information, and integrity
constraints, which only verify the consistency of the knowledge state without changing it.

The setting of multiple ontologies was considered in~\cite{Fang2011}, which considers the
problem of combining information from different knowledge sources while guaranteeing the
overall consistency, and preserving this consistency when one of the individual ontologies
is changed.
This is achieved by external integrity constraints, written in a Datalog-like syntax,
which can refer to knowledge in different ontologies in order to express relationships
between them.
Again, the purpose of these rules is uniquely to identify incompatibilities in the data,
and not to infer new information.

By contrast, the authors who have discussed multi-context systems have not felt the need
to take a similar approach.
Integrity constraints appear routinely in examples in
e.g.~\cite{Brewka2011b,Brewka2011,Eiter2014,Eiter2010}, but always encoded within the
system, so that their violation leads to logical inconsistency of the global knowledge
base.
Their work focuses rather on the aspect of identifying the sources of inconsistencies --
integrity constraints being only one example, not given any special analysis -- and ways
in which it can be repaired.

Although we believe this last work to be of the utmost importance, and show how
satisfaction of integrity constraints can be reduced to consistency checking (which in
turn implies that computing repairs can be reduced to restoring consistency), we strive
for the clean separation between integrity constraints and reasoning that is present in
other formalisms, and believe our proposal to be an important complement to the analysis
of inconsistency in MCSs.

%% file: background.tex
\section{Background}
\label{sec:background}

We begin this section with a summary of the notion of relational multi-context
system~\cite{Fink2011}.
Intuitively, these are a collection of logic knowledge bases -- the
\emph{contexts} -- connected by Datalog-style \emph{bridge rules}.
The formal definition proceeds in several layers.
The first notion is that of \emph{relational logic}, an abstract notion of a logic with a
first-order sublanguage.

\begin{definition}
Formally, a relational logic $L$ is a tuple
$\langle\KB_L,\BS_L,\ACC_L,\Sigma_L\rangle$, where $\KB_L$ is the set
of well-formed knowledge bases of $L$ (sets of well-formed formulas), $\BS_L$ is a set of
possible belief sets (models), $\ACC_L:\KB_L\to 2^{\BS_L}$ is a function assigning to each
knowledge base a set of acceptable sets of beliefs (i.e., its models), and $\Sigma_L$ is a
signature consisting of sets $P^\KB_L$ and $P^\BS_L$ of predicate names (with associated
arity) and a universe $U_L$ of object constants, such that
$U_L\cap(P^\KB_L\cup P^\BS_L)=\emptyset$.
\end{definition}
If $p\in P^\KB_L$ has arity $k$ and $c_1,\ldots,c_k\in U_L$, then
$p(c_1,\ldots,c_k)$ must be an element of some knowledge base, and if
$p\in P^\BS_L$, then $p(c_1,\ldots,c_k)$ must be an element of some
belief set.
Therefore, we can view $\Sigma_L$ as a first-order signature
generating a sublanguage of $L$.
The elements in this sublanguage are called \emph{relational ground
  elements}, while the remaining
elements of knowledge bases or belief sets are called \emph{ordinary}.

\begin{example}
  \label{ex:FOL}
  We can see first-order logic over a first-order signature
  $\Sigma_\FOL$ as a logic
  $\FOL=\langle\KB_\FOL,\BS_\FOL,\ACC_\FOL,\Sigma_\FOL\rangle$, where
  $\KB_\FOL$ is the set of sets of well-formed formulas over
  $\Sigma_\FOL$, $\BS_\FOL$ is the set of first-order interpretations
  over $\Sigma_\FOL$, and $\ACC_\FOL$ maps each set of formulas to the
  set of its models.
  This logic only contains relational elements.
\end{example}

\begin{definition}
Let $\mathfrak I$ be a finite set of indices, $\left\{L_i\right\}_{i\in\mathfrak I}$ be a
set of relational logics, and $V$ be a set of (first-order) variables distinct from
predicate and constant names in any $L_i$.
A \emph{relational element} of $L_i$ has the form $p(t_1,\ldots,t_k)$, where
$p\in P^\KB_{L_i}\cup P^\BS_{L_i}$ has arity $k$ and each $t_j$ is a term from
$V\cup U_{L_i}$, for $1\leq j\leq k$.
A \emph{relational $k$-bridge rule} over $\left\{L_i\right\}_{i\in\mathfrak I}$ and $V$ is
a rule of the form
\begin{equation}
\label{eq:bridge}
(k:s) \leftarrow (c_1:p_1),\ldots,(c_q:p_q), \pnot (c_{q+1}:p_{q+1}),\ldots,\pnot(c_m:p_m)
\end{equation}
such that $k,c_i\in\mathfrak I$, $s$ is an ordinary or a relational knowledge base element
of $L_k$ and $p_1,\ldots,p_m$ are ordinary or relational beliefs of $L_{c_i}$.
\end{definition}
The notation $(c:p)$ indicates that $p$ is evaluated in context $c$.
These rules intuitively generalize logic programming rules, and as
usual in that context we impose a \emph{safety condition}: all
variables occurring in $p_{q+1},\ldots,p_m$ must also occur at least
once in $p_1,\ldots,p_q$.

\begin{definition}
A \emph{relational multi-context system} is a collection
$M=\left\{C_i\right\}_{i\in\mathfrak I}$ of contexts
$C_i=\langle L_i,\kb_i,\br_i,D_i\rangle$, where $L_i$ is a relational logic, $\kb_i$ is a
knowledge base, $\br_i$ is a set of relational $i$-bridge rules, and $D_i$ is a set of
import domains $D_{i,j}$, with $j\in\mathfrak I$, such that $D_{i,j}\subseteq U_j$.
\end{definition}
Import domains define which constants are exported from one context to another: as
the underlying logic languages can be different, these sets are essential to allow one
context to reason about individuals introduced in another.
We will assume that $D_{i,j}$ is the finite domain consisting of the object constants
appearing in $\kb_j$ or in the head of a relational bridge rule in $\br_j$, unless
otherwise stated.

\begin{example}
  \label{ex:mcs}
  Let $C_1$ and $C_2$ be contexts over the first-order logic $\FOL$
  with $\mathsf{R}$ and $\mathsf{Rt}$ binary predicates in
  $\Sigma_\FOL$, and let $\kb_1=\kb_2=\emptyset$.
  We can use the following bridge rules in $\br_2$ to define $\mathsf{Rt}$ in
  $C_2$ as the transitive closure of $\mathsf{R}$ in $C_1$.
  \[(2:\mathsf{Rt}(x,y)) \leftarrow (1:\mathsf{R}(x,y))
    \qquad
    (2:\mathsf{Rt}(x,y)) \leftarrow (1:\mathsf{R}(x,z)),(2:\mathsf{Rt}(z,y))
  \]
  We will use the MCS $M=\langle C_1,C_2\rangle$ to exemplify the
  concepts we introduce.
\end{example}

The semantics of relational MCSs is defined in terms of ground instances of bridge rules:
the instances obtained from each rule $r\in\br_i$ by uniform substitution of each variable
$X$ in $r$ by a constant in $\bigcap D_{i,j}$, with $j$ ranging over the indices of the
contexts to which queries containing $X$ are made in $r$.
\begin{definition}
A \emph{belief state} for $M$ is a collection $S=\left\{S_i\right\}_{i\in\mathfrak I}$
where $S_i\in\BS_i$ for each $i\in\mathfrak I$ -- i.e., a tuple of models, one for each
context.
The ground bridge rule~(\ref{eq:bridge}) is \emph{applicable} in a belief state $S$ if
$p_i\in S_{c_i}$ for $1\leq i\leq q$ and $p_i\not\in S_{c_i}$ for $q<i\leq m$.
The set of the heads of all applicable ground instances of bridge
rules of context $C_i$ w.r.t.\ $S$ is denoted by $\app_i(S)$.
An \emph{equilibrium} is a belief state $S$ such that $S_i\in\ACC_i(\kb_i\cup\app_i(S))$.
\end{definition}
Particular types of equilibria (minimal, grounded, well-founded)~\cite{Brewka2007b} can be
defined for relational MCSs, but we will not discuss them here.

\begin{example}
  \label{ex:equilibria}
  In the setting of the previous example, all equilibria of $M$ will have to
  include the transitive closure of $\mathsf{R}$ in $S_1$ in the
  interpretation of $\mathsf{Rt}$ in $S_2$.
  For example, if we take $S=\langle S_1,S_2\rangle$ with
  $S_1=\{\mathsf{R(a,b)},\mathsf{R(b,c)}\}$ and
  $S_2=\{\mathsf{Rt(a,b)},\mathsf{Rt(b,c)},\mathsf{Rt(a,c)}\}$, then $S$ is an
  equilibrium.
  However, $S'=\langle S_1,S'_2\rangle$ with
  $S'_2=\{\mathsf{Rt(a,b)},\mathsf{Rt(b,c)}\}$ is not an equilibrium, as
  it does not satisfy the second bridge rule.
\end{example}

Checking whether an MCS has an equilibrium is known as the \emph{consistency problem} in
the literature.
We will refer to this property as \emph{logical consistency} (to distinguish from
consistency w.r.t.~integrity constraints, defined in the next section) throughout this
paper.
This problem has been studied
extensively~\cite{Brewka2011b,Eiter2011,Eiter2014,Weinzierl2011}; its decidability depends
on decidability of reasoning in the underlying contexts.
The complexity of checking logical consistency of an MCS $M$ depends on the context
complexity of $M$ -- the highest complexity of deciding consistency in one of the contexts
in $M$ (cf.~\cite{Eiter2014} for a formal definition and known results).

%% file: ic.tex
\section{Integrity constraints on multi-context systems}
\label{sec:ic}

In their full generality, integrity constraints in databases can be arbitrary first-order
formulas, and reasoning with them is therefore undecidable.
For this reason, it is common practice to restrict their syntax in order to regain
decidability; our definition follows the standard approach of writing integrity constraints
in denial clausal form.

\begin{definition}
  Let $M=\langle C_1,\ldots,C_n\rangle$ be an MCS.
  An \emph{integrity constraint} over an MCS $M$ (in denial form) is a formula
  \begin{equation}
    \label{eq:ic}
    \leftarrow(i_1:P_1),\ldots,(i_m:P_m),\pnot(i_{m+1}:P_{m+1}),\ldots,\pnot(i_\ell:P_\ell)
  \end{equation}
  where $M=\langle C_1,\ldots,C_n\rangle$, $i_k\in\{1,\ldots,n\}$, each $P_k$ is a
  relational element of $C_{i_k}$, and the variables in $P_{m+1},\ldots,P_\ell$ all occur
  in $P_1,\ldots,P_m$.
\end{definition}

Syntactically, integrity constraints are similar to ``headless bridge rules''.
However, we will treat them differently: while bridge rules influence the semantics of
MCSs, being part of the notion of equilibrium, integrity constraints are meant to be
checked at the level of equilibria.

\begin{example}
  Continuing the example from the previous section, we can write an
  integrity constraint over $M$ stating that the relation $\mathsf{R}$ (in
  context $C_1$) is transitive.
  \begin{equation}
    \label{eq:ex-ic}
    \leftarrow(2:\mathsf{Rt}(x,y)),\pnot(1:\mathsf{R}(x,y))
  \end{equation}
\end{example}

The restriction on variables again amounts to the usual Logic Programming requirement that
bridge rules be safe.
To capture general tuple-generating dependencies we could relax this
constraint slightly, and allow $P_{m+1},\ldots,P_\ell$ to introduce
new variables, with the restriction that they can be used only once in
the whole rule.
This generalization poses no significant changes to the theory, but
makes the presentation heavier, and we will therefore assume safety.

\begin{definition}
  Let $M=\langle C_1,\ldots,C_n\rangle$ be an MCS and
  $S=\langle S_1,\ldots,S_n\rangle$ be a belief state for $M$.
  Then $S$ satisfies the integrity constraint~\eqref{eq:ic}
  if, for every instantiation $\theta$ of the variables in $P_1,\ldots,P_m$,
  either $P_k\theta\not\in S_k$ for some $1\leq k\leq m$ or $P_k\theta\in S_k$ for
  some $m<k\leq\ell$.
\end{definition}

In other words: equilibria must satisfy all bridge rules (if their body holds, then
so must their heads), but they may or may not satisfy all integrity constraints.
In this sense, integrity constraints express preferences among
equilibria.

\begin{example}
  The equilibrium $S$ from Example~\ref{ex:equilibria} does not
  satisfy the integrity constraint~\eqref{eq:ex-ic}, thus $M$ does not strongly
  satisfy this formula.
  However, $M$ weakly satisfies~\eqref{eq:ex-ic}, as seen by the
  equilibrium $S''=\langle S'_1,S'_2\rangle$ where $S'_2$ is as above and
  $S'_1=\{\mathsf{R(a,b)},\mathsf{R(b,c)},\mathsf{R(a,c)}\}$.
\end{example}

\begin{definition}
  Let $M$ be an MCS and $\eta$ be a set of integrity constraints.
  \begin{enumerate}
  \item $M$ \emph{strongly satisfies} $\eta$, $M\models_s\eta$, if:
    (i)~$M$ is logically consistent and
    (ii)~every equilibrium of $M$ satisfies all integrity constraints in $\eta$.
  \item $M$ \emph{weakly satisfies} $\eta$, $M\models_w\eta$, if there is an equilibrium
    of $M$ that satisfies all integrity constraints in $\eta$.
  \end{enumerate}
\end{definition}
We say that $M$ is (strongly/weakly) \emph{consistent} w.r.t.~a set of integrity
constraints $\eta$ if $M$ (strongly/weakly) satisfies $\eta$.
These two notions express different interpretations of integrity constraints.
Strong satisfaction views them as necessary requirements, imposing that all models of the
MCS to satisfy them.
Examples of these are the usual integrity constraints over databases, which express
semantic connections between relations that must always hold.
Weak satisfaction views integrity constraints as expressing preferences: the MCS may have
several equilibria, and we see those that do satisfy the integrity constraints as
``better''.

The distinction is also related to the use of brave (credulous) or cautious (skeptical)
reasoning.
If $M$ strongly satisfies a set of integrity constraints $\eta$, then any inferences we
draw from $M$ using brave reasoning are guaranteed to hold in some equilibrium that also
satisfies $\eta$.
If, however, $M$ only weakly satisfies $\eta$, then this no longer holds, and we can only
use cautious reasoning if we want to be certain that any inferences are still compatible
with $\eta$.

Both strong and weak satisfaction require $M$ to be logically consistent, so
$M\models_s\eta$ implies $M\models_w\eta$.
This implies that deciding whether $M\models_s\eta$ and $M\models_w\eta$ are both at least
as hard as deciding whether $M$ has an equilibrium -- thus undecidable in the general
case.\footnote{If consistency of one of $M$'s contexts is undecidable, then clearly the
  question of whether $M$ has an equilibrium is also undecidable.}
When logical consistency of $M$ is decidable and its set of equilibria is enumerable, weak
satisfaction is semi-decidable (if there is an equilibrium that satisfies $\eta$, we
eventually encounter it), while strong satisfaction is co-semi-decidable (if there is an
equilibrium that does not satisfy $\eta$, we eventually encounter it).
The converse also holds.

\begin{theorem}
  \label{thm:wdec}
  Weak satisfaction of integrity constraints is reducible to logical consistency.
\end{theorem}
\begin{proof}
  To decide whether $M\models_w\eta$,
  construct $M'$ by extending $M$ with a context $C_0$ where $\KB_0=\wp(\{\ast\})$,
  $\kb_0=\emptyset$, $\ACC_0(\emptyset)=\{\emptyset\}$, $\ACC_0(\{\ast\})=\emptyset$, and
  the bridge rules obtained by adding $(0:\ast)$ to the head of the rules in $\eta$.
  Then $M'$ has an equilibrium iff $M\models_w\eta$: any equilibrium
  of $M$ not satisfying $\eta$ corresponds to a belief state of
  $M'$ where $\app_0(S)=\{\ast\}$, which is never an equilibrium
  of $M'$; but equilibria of $M$ satisfying $\eta$ give rise to
  equilibria of $M'$ taking $S_0=\emptyset$.
  \qed
\end{proof}

\begin{theorem}
  \label{thm:sdec}
  Strong satisfaction of integrity constraints is reducible to logical inconsistency.
\end{theorem}
\begin{proof}
  Construct $M'$ as before, but now defining $\ACC_0(\emptyset)=\emptyset$,
  $\ACC_0(\{\ast\})=\{\{\ast\}\}$.
  If $M$ is inconsistent, then $M\not\models_s\eta$.
  If $M$ is consistent, then any equilibrium of $M$ satisfying $\eta$ corresponds to a
  belief state of $M'$ where $\app_0(S)=\emptyset$, which can never be an
  equilibrium of $M'$; in turn, equilibria of $M$ not satisfying $\eta$
  give rise to equilibria of $M'$ taking $S_0=\{\ast\}$.
  So if $M$ is consistent, then $M\models_s\eta$ iff $M'$ is inconsistent.
  \qed
\end{proof}

Combining the two above results with the well-known complexity results for consistency checking (Table 1 in \cite{Eiter2014}), we directly obtain the following results.

\begin{corollary}
The complexity of deciding whether $M \models_w \eta$ or $M \models_s \eta$, depending on the context complexity of $M$, $\mathcal{CC}(M)$, is given in Table~\ref{tab:complexity}.
\end{corollary}

\begin{table}[t]
\begin{center}
\begin{tabular}{r@{\hskip1em}|@{\hskip1em}c@{\hskip1em}c@{\hskip1em}c@{\hskip1em}c@{\hskip1em}c}
$\mathcal{CC}(M)$ & P & NP & $\Sigma^p_i$ & PSPACE & EXPTIME\\
\hline
$M \models_w \eta$ & NP & NP & $\Sigma^p_i$ & PSPACE & EXPTIME\\
$M \models_s \eta$ & $\Delta^p_2$ & $\Delta^p_2$ & $\Delta^p_{i+1}$ & PSPACE & EXPTIME\\
\end{tabular}
\end{center}
\caption{Complexity of integrity checking of an MCS in terms of its context complexity.}
\label{tab:complexity}
\end{table}

These results suggest an alternative way of modelling integrity constraints in MCSs:
adding them as bridge rules whose head is a special atom interpreted as
inconsistency.
This approach was taken in e.g.~\cite{Eiter2011}.
However, we believe that integrity constraints should be kept separate from the data, and
having them as a separate layer achieves this purpose.
In this way, we do not restrict the models of MCSs, and we avoid issues of logical
inconsistency.
Furthermore, violation of integrity constraints typically is indicative of some error in
the model or in the data, which should result in an alert and not in additional
inferences.

These considerations are similar to those made in Section~2.7 of~\cite{Motik2009} and in \cite{Fang2011}, in the (more restricted) context of integrity constraints over description logic knowledge bases. Likewise, the approach taken for integrity constraints in databases is that inconsistencies should be brought to the users' attention, but not affect the semantics of the database \cite{Abiteboul1988,Eiter1992}. In particular, it may be meaningful to work with reasoning systems not satisfying integrity constraints (see \cite{Chomicki2010} for databases and \cite{Eiter2010} for description logic knowledge bases). Our approach is also in line with \cite{Brewka2011}, where it is argued that in MCSs it is important to ``distinguish data from additional operations on it''.

%% file: apps.tex
\section{Applications of ICs for MCSs}
\label{sec:apps}

In this section we look at particular cases of MCSs with integrity
constraints.
We begin by showing that our notion generalizes the usual one for
standard databases.
Then we look into other types of databases and show how we obtain
interesting notions for these systems.

\subsection{Relational databases}

Integrity constraints in relational databases can be written as first-order
formulas in denial clausal form~\cite{Flesca2004} -- which are
essentially equivalent in form to bridge rules with no head.

\begin{definition}
  Let $DB$ be a database.
  The context generated by $DB$, $\Ctx(DB)$, is defined as follows.
  \begin{itemize}
  \item The underlying logic is first-order logic.
  \item Belief sets are sets of ground literals.
  \item The knowledge base is $DB$.
  \item For all $\kb$, the only belief set compatible with $\kb$ is
    $\ACC(\kb)=\kb^\vdash=\kb\cup\{\neg a\mid a\not\in\kb\}$.
  \item The set of bridge rules is empty.
  \end{itemize}
\end{definition}
We can see any database $DB$ as a single-context MCS consisting of
exactly the context $\Ctx(DB)$; we will also denote this MCS by
$\Ctx(DB)$, as this poses no ambiguity.
The only equilibrium for $\Ctx(DB)$ is $DB^\vdash$ itself,
corresponding to the usual closed-world semantics of relational databases.
Previous work (cf.~\cite{Brewka2011,Eiter2014}) implicitly treats
databases in this way, although $\Ctx$ is not formally defined.

Let $DB$ be a database and $r$ be an integrity constraint over $DB$ in
denial clausal form.
We can rewrite $r$ as an integrity constraint over $\Ctx(DB)$: if $r$
is
$\forall(A_1\wedge\ldots\wedge A_k\wedge\neg B_1\wedge\ldots\wedge\neg B_m\to\bot)$,
then $\br(r)$ is
\[\leftarrow(1:A_1),\ldots,(1:A_k),\pnot(1:B_1),\ldots,\pnot(1:B_m)\,.\]
Note that general tuple-generating dependencies require allowing singleton variables in
the $B_i$s, as discussed earlier.
The following result is straightforward to prove.
If we assume first-order logic with equality, we can also write equality-generating
constraints, thus obtaining the expressivity used in databases.

\begin{theorem}
  \label{thm:db}
  Let $DB$ be a database and $\eta$ be a set of ICs over $DB$.
  Then $DB$ satisfies all ICs in $\eta$ iff
  $\Ctx(DB)\models_s\br(\eta)$ iff $\Ctx(DB)\models_w\br(\eta)$, where $\br$ is
  extended to sets in the standard way.
\end{theorem}

In this setting, weak and strong satisfaction of integrity constraints
coincide, as every database has exactly one equilibrium.
Furthermore, deciding whether $\Ctx(DB)\models\br(\eta)$ can be done
in time $O(|DB|\times|\eta|)$, where $|DB|$ is the number of elements
in $DB$ and $|\eta|$ is the total number of literals in all integrity
constraints in $\eta$. This means that the data complexity~\cite{Vardi1982} of this
problem is linear, as we can query the database using the open bridge
rules in $\eta$, rather than considering the set of all ground instances of
those rules.

Theorem~\ref{thm:db} could be obtained by adding integrity constraints as bridge rules
with a special inconsistency atom, as discussed earlier, and done in~\cite{Eiter2011}).
This would significantly blur the picture, though, as in principle nothing would
prevent us from writing integrity constraints referencing the inconsistency atom
in their body, potentially leading to circular reasoning. Our approach guarantees
that there is no such internalization of inconsistencies into the database.

Our results show that the notion of integrity
constraint we propose directly generalizes the traditional notion of
integrity constraints over databases~\cite{Abiteboul1988}.

\subsection{Distributed DBs}
\label{sec:distr-db}

Distributed databases are databases that store their information at
different sites in a network, typically including information that is
duplicated at different nodes~\cite{Ullman1988} in order to promote
resilience of the whole system.

A distributed database consisting of individual databases
$DB_1,\ldots,DB_n$ can be modeled as an MCS with $n$ contexts
$\Ctx(DB_1),\ldots,\Ctx(DB_n)$.
The internal consistency of the database, in the sense that tables
that occur in different $DB_i$s must have the same rows, can be
specified as integrity constraints over this MCS as follows.
For each relation $p$, let $\gamma(p)$ be the number of columns of $p$
and $\delta(p)$ be the set of indices of the databases containing $p$.
Then
\[\{\leftarrow(i:p(x_1,\ldots,x_{\gamma(p)})), \pnot(j:p(x_1,\ldots,x_{\gamma(p)}))\mid i,j\in\delta(p), \mbox{$p$ is a relation}\}\]
logically specifies the integrity of the system.
Different strategies for fixing inconsistencies in distributed
databases (e.g.~majority vote or siding with the most recently updated
node) correspond to different preferences for choosing repairs in the
sense of the next section.

Again, such integrity constraints can be written as bridge rules in the form
\[(j:p(x_1,\ldots,x_{\gamma(p)})) \leftarrow (i:p(x_1,\ldots,x_{\gamma(p)}))\,.\]
but these significantly change the semantics of the database: instead of
describing preferred equilibria, they impose a flow of information between nodes. 


\begin{example}
\label{ex:cpr}
Consider a country with a central person register (CPR), mapping a unique identifying
number to the name and current address of each citizen using a relation ${\sf person}$,
e.g.\ ${\sf person}(1111111118,{\it old\_lady},{\it gjern})$.
Furthermore, each electoral district keeps a local voter register using a relation
${\sf voter}$, e.g.\ ${\sf voter}(1111111118)$, and a list of addresses local to the given
electoral district using a relation ${\sf address}$,
e.g.~${\sf address}({\it gjern})$.
Then the integrity constraints 
\begin{align}
&\leftarrow {\sf Skborg}: {\sf voter}({\it Id}), \pnot({\sf CPR}: {\sf person}({\it Id}))\\
&\leftarrow {\sf Skborg}: {\sf voter}({\it Id}), {\sf CPR}: {\sf person}({\it Id},{\it Add}), \pnot({\sf Skborg} : {\sf address}({\it Add})) \label{ic1}
\end{align}
ensure that all voters registered in the Silkeborg electoral district are registered in
the central person register, and that they are registered with an address that is local to
the Silkeborg electoral district.
Here, we are implicitly assuming that the database is closed under projection, and
overload the {\sf person} relation for the sake of simplicity.
In addition, the following set of integrity constaints models the fact that each person
registered in the Silkeborg electoral district is not registered in any other electoral
districts from the set ${\it ED}$.
\[\{\leftarrow {\sf Skborg}: {\sf voter}({\it Id}), C_i : {\sf voter}({\it Id}) \mid C_i \in{\it ED} \setminus \{ {\sf Skborg} \} \}\]
\end{example}

This assumption of closure under projection is meaningful from a
practical point of view, and has been implemented e.g.\ in~\cite{KMIS2015}.
Alternatively, we could define the projections as bridge rules of the MCSs, in line with
the idea of encoding views of deductive databases presented in the next section.

This section's treatment of distributed databases is equivalent to considering their
disjoint union as a database.
Consequently, there is no need to use MCSs for distributed databases, but this mapping
shows that our notion of integrity constraints abstracts the practice in this field.
Furthermore, results in previous work~\cite{LCF2014} indicate that the processing of
integrity constraints can be efficiently parallelized in this disjoint scenario, given
suitable assumptions.

\subsection{Deductive DBs}

We now address the case of deductive databases.
These consist of two different components: the (extensional) fact
database, containing only concrete instances of relations, and the
(intensional) rule database, containing Datalog-style rules defining
new relations.
Every relation must be either intensional or extensional, unlike in
e.g.~full-fledged logic programming.

One standard way to see the intensional component(s) of deductive
databases is as \emph{views} of the original database.
The instances of the new relations defined by rules are generated
automatically from the data in the database, and these
relations can thus be seen as content-free, having a purely
presentational nature.
For simplicity of presentation, we consider the case where there
is one single view.

\begin{definition}
  Let $\Sigma_E$ and $\Sigma_I$ be two disjoint first-order signatures.
  A \emph{deductive database} over $\Sigma_E$ and $\Sigma_I$ is a pair
  $\langle DB,R\rangle$, where $DB$ is a relational database over
  $\Sigma_E$ and $R$ is a set of rules of the form
  $p\leftarrow q_1,\ldots,q_n$,
  where $p$ is an atom of $\Sigma_I$ and $q_1,\ldots,q_n$ are atoms
  over $\Sigma_E\cup\Sigma_I$.
\end{definition}
More precisely, this definition corresponds to the definite deductive databases
in~\cite{Gallaire1984}; we do not consider the case of indefinite databases in this work.
We can view deductive databases as MCSs.

\begin{definition}
  Let $\langle DB,R\rangle$ be a deductive database over $\Sigma_E$
  and $\Sigma_I$.
  The MCS induced by $\langle DB,R\rangle$ is
  $M=\langle C_E,C_I\rangle$, where $C_E=\Ctx(DB)$ defined as
  above and $C_I=\Ctx(R)$ is a similar context where:
  \begin{itemize}
  \item The knowledge base is $\emptyset$.
  \item For each rule $p\leftarrow q_1,\ldots,q_n$ in $R$ there is a
    bridge rule $(I:p)\leftarrow (i_1:q_1),\ldots,(i_n:q_n)$ in $\Ctx(R)$,
    where $i_k=E$ if $q_k$ is an atom over $\Sigma_E$ and $i_k=I$ otherwise.
  \end{itemize}
\end{definition}

Integrity constraints over such MCSs correspond precisely to the definition of integrity
constraints over deductive databases from~\cite{Asirelli1985}.
By combining this with the adequate notion of repair, we capture the typical constraints
of deductive databases -- that consistency can only be regained by changing extensional
predicates -- in line with the traditional view-update problem.
More modern works~\cite{Caroprese2012} restrict the syntax of integrity constraints,
allowing them to use only extensional relations; in the induced MCS, this translates to
the additional requirement that only relational elements from $C_E$ appear in the body
of integrity constraints.

\begin{example}
Consider a deductive database for class diagrams, where information about direct
subclasses is stored in the extensional database using a relation ${\sf isa}$,
e.g.~${\sf isa}({\it list},{\it collection})$ and ${\sf isa}({\it array},{\it list})$.
Intensionally, we model the transitive closure of the subclass relation using a view
created by the two rules ${\sf sub}(A,B) \leftarrow {\sf isa}(A,B)$ and
${\sf sub}(A,C) \leftarrow {\sf isa}(A,B), {\sf sub}(B,C)$, thus allowing us to find out
that in our example ${\sf sub}({\it array},{\it collection})$.
The integrity constraint
\[\leftarrow{\sf sub}(A,A)\]
can then be used to state the acyclicity of the subclass relation.
Integrity constraints restricted to the extensional database could not express
this, as there would be no way to define a fixpoint.
The only (incomplete) solution would be to add $n$ integrity constraints
disallowing cycles of length up to $n$.
This example illustrates our gain of expressive power compared to the approach
in~\cite{Caroprese2012}.
\end{example}

We can also consider databases with several, different views, each view generating a
different context.
Integrity constraints over the resulting MCS can then specify relationships between
relations in different views.

Yet again, the complexity of verifying whether an MCS induced by a deductive database
satisfies its integrity constraints is lower than the general case.
In particular, consistency checking is reducible to query answering (all integrity
constraints are satisfied iff there are no answers to the queries expressed in their
bodies).
If we do not allow negation in the definition of the intensional relations, then there is
only one model of the database as before, and consistency checking w.r.t.\ a fixed set of integrity constraints
is PTIME-complete \cite{Schlipf1995}.
In the general case, weak and strong consistency correspond, respectively, to brave and
cautious reasoning for Datalog programs under answer set semantics, which are known to
be co-NP-complete and NP-complete, respectively.

\subsection{Peer-to-peer systems}
\label{sec:p2p}

Peer-to-peer (P2P) networks are distributed systems where each node (the peer) has an identical status in the hierarchy, i.e., there is no
centralized control. Queries can be posed to each peer, and peers communicate amongst themselves in order to produce the desired answer. For a general overview see e.g.~\cite{Pourebrahimi05asurvey}.

A particularly interesting application are P2P systems, which integrate features of both
distributed and deductive databases.
We follow~\cite{Caroprese2007}, which also addresses the issue of integrity constraints.
In this framework, P2P systems consist of several nodes (the peers), each of them a
deductive database of its own, connected via \emph{mapping rules} that port
relations from one peer to another.

\begin{definition}
  A peer-to-peer system $\mathcal P$ is a set of peers
  $\mathcal P=\{P_i\}_{i=1}^n$.
  Each peer is a tuple $\langle \Sigma^i,DB_i,R_i,M_i,IC_i\rangle$, where:
  \begin{itemize}
  \item $\Sigma^i$ is the disjoint union of three signatures
    $\Sigma^i_E$, $\Sigma^i_I$ and $\Sigma^i_M$;
  \item $\langle DB_i,R_i\rangle$ is a deductive database over signatures
    $\Sigma^i_E$ and $\Sigma^i_I$, where the rules in $R_i$ may also
    use relations from $\Sigma^i_M$;
  \item $M_i$ is a set of mapping rules of the form
    $p\leftarrow_j q_1,\ldots,q_m$ with $j\neq i$, where $p$
    is an atom over a signature $\Sigma^i_M$ and each $q_k$ is an atom
    over $\Sigma^j$;
  \item $IC_i$ is a set of integrity constraints over $\Sigma^i$.
  \end{itemize}
\end{definition}
Intuitively, relations can be defined either extensionally (those in
$\Sigma_E$), intensionally (those in $\Sigma_I$) or as mappings from
another peer (those in $\Sigma_M$), and these definitions may not be
mixed.
Observe that, with these definitions, negations may only occur in the bodies of the
integrity constraints.

We can view a P2P system as a MCS with integrity constraints.
To simplify the construction, we adapt the definition from the case of deductive databases
slightly, so that there is a one-to-one correspondence between peers and contexts.

\begin{definition}
  Let $\mathcal P=\{P_i\}_{i=1}^n$ be a P2P system.
  The MCS induced by $\mathcal P$ is defined as follows.
  \begin{itemize}
  \item There are $n$ contexts, where $C_i$ is constructed as
    $\Ctx(DB_i)$ together with the following set of bridge rules:
    \begin{itemize}
    \item $(i:p)\leftarrow(i:q_1),\ldots,(i:q_m)$ for each rule
      $p\leftarrow q_1,\ldots,q_m\in R_i$;
    \item $(i:p)\leftarrow(j:q_1),\ldots,(j:q_m)$ for each rule
      $p\leftarrow_j q_1,\ldots,q_m\in M_i$.
    \end{itemize}
  \item Each integrity constraint $\leftarrow q_1,\ldots,q_m$ in
    $IC_i$ is translated to the integrity constraint
    $\leftarrow(i:q_1),\ldots,(i:q_m)$, where we take $(i:\neg q)$ to
    mean $\pnot(i:q)$.
  \end{itemize}
\end{definition}
The definition of the bridge rules from $R_i$ is identical to what one would
obtain by constructing the context $\Ctx(R_i)$ described in the
previous section.

This interpretation does not preserve the semantics for
P2P systems given in~\cite{Caroprese2007,Caroprese2014}.
Therein, mapping rules can only be applied if
they do not generate violations of the integrity constraints.
This is directly related to the real-life implementation of these systems, where this
option represents a ``cheap'' strategy to ensure local enforcement of integrity
constraints; as discussed in~\cite{Weinzierl2011}, the underlying philosophy of P2P
systems and MCSs is significantly different.

We now show that, while the semantics differ, there is a correspondence between
P2P systems and their representation as an MCS, and the ``ideal'' models of both coincide.
When no such models exist, the MCS formulation can be helpful in identifying the
problematic mapping rules.

The semantics of P2P systems implicitly sees them as logic programs.
\begin{definition}
  Let $\mathcal P=\{P_i\}_{i=1}^n$ be a P2P system and $I$ be a Herbrand interpretation
  over $\bigcup\Sigma^i$.
  The program $\mathcal P^I$ is obtained from $\mathcal P$ by (i)~grounding all rules and
  (ii)~removing the mapping rules whose head is not in $I$.

  Let $\mathcal{MM}(P)$ denote the minimal model of a logic program.
  A \emph{weak model} for $\mathcal P$ is an interpretation $I$ such that
  $I=\mathcal{MM}(\mathcal P^I)$.
\end{definition}
Since integrity constraints are rules with empty head, this definition implicitly requires
weak models to satisfy them.
Interpretations over a P2P system and equilibria over the induced MCS are trivially in
bijection, as the latter simply assign each atom to the right context, and we implicitly
identify them hereafter.
We can relate the ``perfect'' models in both systems.

\begin{theorem}
  Let $\mathcal P$ be a P2P system, $I$ an interpretation for $\mathcal P$, and
  $M$ the induced MCS.
  Then $I=\mathcal{MM}(\mathcal P)=\mathcal{MM}(\mathcal P^I)$ iff $I$ is an equilibrium
  for $M$ satisfying all the integrity constraints.
\end{theorem}
\begin{proof}
  Since $\mathcal P$ corresponds to a positive program, the only equilibrium of $M$ is
  $\mathcal{MM}(\mathcal P)$ (see~\cite{CHN2013b}).
  Furthermore, for any $I$, $\mathcal{MM}(\mathcal P^I)$ includes the facts in all
  extensional databases and satisfies all rules in $R_i$ and all integrity constraints.
  Thus, it also corresponds to a belief state satisfying their counterparts in $M$.

  Suppose that $\mathcal{MM}(\mathcal P)=\mathcal{MM}(\mathcal P^I)$.
  Since mapping rules are the only ones that can add information about relations in
  $\Sigma^i_M$ to $I$, the second equality implies that no mapping rules are removed in
  $\mathcal P^I$.
  Therefore $I=\mathcal{MM}(\mathcal P)$ satisfies all bridge rules of $M$ obtained from
  the mapping rules in $\mathcal P$, whence $I$ is an equilibrium of $M$ satisfying all
  integrity constraints.

  Conversely, if $I$ is an equilibrium of $M$ and $r$ is a mapping rule, then either $I$
  does not satisfy the body of $r$ or $I$ contains its head.
  Since no other rules can infer instances of relations in $\Sigma^i_M$, this implies that
  $\mathcal{MM}(\mathcal P)=\mathcal{MM}(\mathcal P^I)$, and being an equilibrium implies
  that $I=\mathcal{MM}(\mathcal P)$.
\qed
\end{proof}

The MCS representation has an interesting connection with the notion of weak model in
general, though: if there are integrity constraints in $M$ that are not satisfied by
$\mathcal{MM}(\mathcal P)$, then repairing $M$ by removing mapping rules is equivalent to
finding a weak model for $\mathcal P$.
This is again reminescent of the view-update problem.

The MCS representation allows us to write seemingly more powerful integrity constraints
over a P2P system, as we can use literals from different contexts in the same rule.
However, this does not give us more expressive power: for example, the integrity constraint
$\leftarrow(1:a),(2:b)$
can be written as
$\leftarrow(1:a),(1:b_2)$
adding the mapping rule $(1:b_2)\leftarrow(2:b)$, where $b_2$ is a fresh relation in peer
$1$.

%

\subsection{Description Logic Knowledge Bases}

We now discuss the connection between our work and results on adding
integrity constraints to description logic knowledge bases, namely OWL ontologies.

Description logics differ from databases in their rejection of the closed-world
assumption, thereby contradicting the semantics of negation-by-failure.
For this reason, encoding ontologies as a context in an MCS is a bit different than the
previous examples.
We follow the approach from~\cite{CGN2014}, refering the reader to the discussion therein
of why the embeddding from e.g.~\cite{Brewka2007b} is not satisfactory.

\begin{definition}
  A description logic $\mathcal L$ is represented as the relational logic
  $L_{\mathcal L}=\langle\KB_{\mathcal L},\BS_{\mathcal L},\ACC_{\mathcal L},\Sigma_{\mathcal L}\rangle$
  defined as follows:
  \begin{itemize}
  \item $\KB_{\mathcal L}$ contains all well-formed knowledge bases
    (including a T-Box and an A-Box) of~$\mathcal L$;
  \item $\BS_{\mathcal L}$ is the set of all possible A-Boxes in the language of $\mathcal L$;
  \item $\ACC_{\mathcal L}(\kb)$ is the singleton set containing the set of $\kb$'s known
    consequences (positive and negative);
  \item $\Sigma_{\mathcal L}$ is the signature underlying~$\mathcal L$.
  \end{itemize}
\end{definition}
Regarding the choice of acceptable belief sets (the elements of
$\BS_{\mathcal L}$), the possible A-Boxes
correspond to (partial) models of $\mathcal L$,
seen as a first-order theory: they contain concepts and roles applied
to particular known individuals, or negations thereof.
However, they need not be categorical:
they may contain neither $C(a)$ nor $\neg C(a)$ for particular $C$ and $a$.
This reflects the typical open-world semantics of ontologies.
In particular, the only element of $\ACC_{\mathcal L}(\kb)$ may not be a model of
$\kb$ in the classical sense of first-order logic.
This is in contrast with~\cite{Brewka2007b}, where $\ACC_{\mathcal L}(\kb)$ contains all
models of $\kb$; as discussed in~\cite{CGN2014}, this is essential to model e.g.~default
reasoning correctly.

\begin{definition}
  An ontology $\mathcal O$ based on description logic $\mathcal L$ induces a context with
  underlying logic $L_{\mathcal L}$, knowledge base $\mathcal O$, and an empty set of
  bridge rules.
\end{definition}

Like in the database scenario, ontologies viewed as MCSs always have one equilibrium, as
long as they are logically consistent.
Therefore, the notions of weak and strong satisfaction of integrity constraints again
coincide, and we get the same notion of consistency w.r.t.~a set of integrity constraints
as that defined in~\cite{Motik2009}; however, our syntax is more restricted, as we do not
allow general formulas as integrity constraints.
Observe that, as in that work, our integrity constraints only apply to named individuals
(explicitly mentioned in the ontology's A-Box), which is a desirable consequence that yet
again can only be gained from keeping integrity constraints separate from the knowledge
base.

\begin{example}
  We illustrate the construction in this section with a classical
  example.
  We assume that we have an ontology $O$ including a concept
  $\mathsf{person}$ and a role $\mathsf{hasCPR}$, which associates
  individuals with their CPR number.
  (So we are essentially resetting Example~\ref{ex:cpr} to use an
  ontology, rather than a distributed database.)
  We can add the integrity constraint
  \[\leftarrow(O:\mathsf{person}(x)),\pnot(O:\mathsf{hasCPR}(x,y))\]
  requiring each person to have a CPR number.
  Due to the semantics of ontologies, this actually requires each
  person's CPR number to be explicitly present in the ontology: the
  presence of an axiom such as
  $\mathsf{person}\sqsubseteq(\exists\mathsf{person.hasCPR})$ does not
  yield any instance $\mathsf{hasCPR}(x,y)$ in the set of the
  ontology's known consequences.
  This also justifies our definition of $\ACC_{\mathcal L}$: if we
  take the model-based approach of~\cite{Brewka2007b}, then this
  integrity constraint no longer demands the actual presence of such a
  fact in the A-Box.

  This integrity constraint is an example of one that does not satisfy
  the safety condition (the variable $y$ occurs only in a negated
  literal), but as discussed in Section~\ref{sec:ic} our theory is
  easily extended to cover this case, as $y$ only occurs once in the formula.
\end{example}

Our scenario is also expressive enough to model the distributed ontology scenario
of~\cite{Fang2011}, which defines integrity constraints as logic
programming-style rules with empty head whose body can include atoms from different
ontologies: we can simply consider the MCS obtained from viewing each ontology as a
separate context, and the integrity constraints as ranging over the
joint system.

%% file: repair.tex
\section{Repairs and managed multi-context systems}
\label{sec:repair}

The definitions in the previous section allow us to distinguish between acceptable and
non-acceptable equilibria w.r.t.\ a set of integrity constraints, but they do not help
with the analog of the problem of database repair~\cite{Abiteboul1988} -- namely, given an
inconsistent equilibrium for a given MCS, how do we change it into a consistent one.
In order to address this issue, we turn our attention to
\emph{managed} multi-context systems (mMCS)~\cite{Brewka2011}.

\begin{definition}
  A \emph{managed multi-context system} is a collection of managed contexts
  $\{C_i\}_{i\in\mathcal J}$, with each
  $C_i=\langle L_i,\kb_i,\br_i,D_i,OP_i,\mng_i\rangle$ as follows.
  \begin{itemize}
  \item $L_i$ is a relational logic, $\kb_i$ is a knowledge base, and
    $D_i$ is a set of import domains, as in standard MCSs.
  \item $OP_i$ is a set of operation names.
  \item $\br_i$ is a set of managed bridge rules, with the form of
    Equation~\eqref{eq:bridge}, but where $s$ is of the form $o(p)$
    with $o\in OP_i$ and $p\in\bigcup\KB_i$.
  \item $\mng_i:\wp(OP_i\times\bigcup\KB_i)\times\KB_i\to\KB_i$ is a
    \emph{management function}.
  \end{itemize}
\end{definition}

The intuition is as follows: the heads of bridge rules can now contain
arbitrary actions (identified by the labels in $OP_i$, and the
management function specifies the semantics of these labels --
see~\cite{Brewka2011} for a more detailed discussion.
Our definition is simplified from those authors', as they allow the
management function to change the semantics of the contexts and return
several possible effects for each action.
This simplification results in a less flexible concept of mMCS, which is however more useful for the
purposes of defining repairs.

\begin{example}
  The management function can perform several
  manipulations of the knowledge base in one update action.
  For example, considering the setting of Example~\ref{ex:cpr}, we could include
  an operation ${\sf replace}\in OP_{\sf CPR}$ such that
  $\mng(\{\langle{\sf replace},{\sf person}({\it Id},{\it Name},{\it Add})\rangle\},\kb)$
  inserts the tuple $({\it Id},{\it Name},{\it Add})$ into the ${\sf person}$ table and
  removes any other tuple $({\it Id},{\it Name}',{\it Add}')$ from that table.
\end{example}

Every MCS (in the sense of the previous section) can be seen as an
mMCS by taking every context to have exactly one operation
${\sf add}$ with the natural semantics of adding its argument (the
head of the rule) to the belief set associated with the context in
question.
We will therefore discuss integrity constraints over mMCS in the
remainder of this section.
The motivation of generalizing database tradition also suggests that
we include another operation ${\sf remove}$ that removes an element
from the specified context.

\begin{definition}
  Let $M=\{C_i\}_{i\in\mathcal I}$ be an mMCS.
  An \emph{update action} for $M$ is of the form
  $(i:o(p))$, with $i\in\mathcal J$, $o\in OP_i$ and $p\in\bigcup\KB_i$.

  Given a set of update actions $\U$ and an mMCS $M$, the
  result of applying $\U$ to $M$, denoted $\U(M)$,
  is computed by replacing each $\kb_i$ (in context $C_i$) by
  $\mng_i(\U_i,\kb_i)$, where $\U_i$ is the set of update actions of
  the form $(i:o(p))$.
\end{definition}

Updates differ from applying (managed) bridge rules, as they actually change one or more
knowledge bases in $M$'s contexts \emph{before} any evaluation of bridge rules
takes place.
This is similar to database updates, which change the database before and independent of
the query processing.
Based on this notion of update, we can define (weak) repairs as follows.

\begin{definition}
  Let $M$ be an mMCS, $\eta$ be a set of ICs over $M$, and assume that $M$ is inconsistent
  w.r.t.~$\eta$.
  A set of update actions $\U$ is a \emph{weak repair} for $M$ and $\eta$ if $\U(M)$ is
  consistent w.r.t.~$\eta$.
  If there is no subset $\U'$ of $\U$ that is also a weak repair for $M$ and $\eta$, then
  $\U$ is a \emph{repair}.
\end{definition}

\begin{example}
  Again in the setting of Example~\ref{ex:cpr}, suppose that the CPR database contains the record
  ${\sf person}(1111111118,{\it old\_lady},{\it odense})$ and the Silkeborg electoral
  database contains the records ${\sf voter}(1111111118)$ and
  ${\sf address}({\it gjern})$, but not the record ${\sf address}({\it odense})$ as Odense
  is not in Silkeborg.
  The induced mMCS is inconsistent w.r.t.~the integrity constraint~\eqref{ic1},
  and a possible repair is
  $\{({\sf CPR}:{\sf add}({\sf person}(1111111118,{\it old\_lady},{\it gjern})))\}$.
  The semantics of the management function guarantee that only the new record will persist
  in the mMCS.
\end{example}

As is the case in databases, it can happen that a set of integrity constraints is
inconsistent, in the sense that no MCS can satisfy it.
However, this inconsistency can also arise from incompatibility between integrity
constraints and bridge rules -- consider the very simple case where there is a bridge rule
$(B:{\sf b}) \leftarrow (A:{\sf a})$ and an integrity constraint
$\leftarrow(A:{\sf a}),\pnot(B:{\sf b})$).
Since our notion of update does not allow one to change bridge rules, this inconsistency
is unsurmountable.

In general, this interaction between integrity constraints and bridge rules makes the
problem of finding repairs for inconsistent MCSs more complex than in the database world.
However, Theorems~\ref{thm:wdec} and~\ref{thm:sdec} show that the problem of finding a
repair for an MCS that is inconsistent w.r.t.~a set of integrity constraints can be
reduced to finding a set of update actions that will make a logically inconsistent MCS
have equilibria.
The results on diagnosing and repairing logical inconsistency in multi-context
systems~\cite{Eiter2011,Eiter2014} can therefore be used to tackle this problem.
By considering deductive databases as MCSs, we also see the problem of repairing an
inconsistent MCS as a generalization of the view-update
problem~\cite{Kakas1990,Mayol2003,Teniente1995}.

Another issue is how to choose between different repairs: as in the database case, some
repairs are preferable to others.
Consider the following toy example.
\begin{example}
  Let $M$ be the MCS induced by a deductive database with one extensional relation $\sf p$
  and one intensional relation $\sf q$, both $0$-ary, connected by the rule
  $\sf q\leftarrow\sf p$, and consider the integrity constraint $(I:\sf q)$.

  Assume the usual operations ${\sf add}$ and ${\sf remove}$.
  There are two repairs for $M$, namely $\{(E:\sf add(\sf p))\}$ and
  $\{(I:\sf add(\sf q))\}$, but only the former is valid from the perspective of deductive
  databases.
\end{example}

The usual consensus in databases is that, in general, deciding which repair to apply is a
task that needs human intervention~\cite{Eiter1992}.
However, several formalisms also include criteria to help automate such preferences.
In our setting, a simple way to restrict the set of possible repairs would be to restrict
the update actions to use only a subset of the $OP_i$s -- in the case of deductive databases,
we could simply restrict them to the operations over $C_E$.
An alternative that offers more fine-tuning capabilities would be to go in the direction
of active integrity constraints~\cite{Flesca2004}, which require the user to be explicit
about which update actions can be used to repair the integrity constraints that are not
satisfied.
We plan to pursue the study of such formalisms to discuss repairs of MCSs with integrity
constraints in future work.
We also intend to study generalizations of repairs to include the possibility of changing
bridge rules.

%% file: conclusion.tex
\section{Conclusions and Future Work}
\label{sec:conclusion}

In this paper, we proposed a notion of integrity constraint for multi-context systems, a
general framework for combining reasoning systems.
We showed that our notion generalizes the well-studied concept of integrity constraint
over databases, and studied its relation to similar notions in other formalisms.
Satisfaction of integrity constraints comes in two variants, weak and strong, related to
the usual concepts of brave and cautious reasoning.

By showing how to encode integrity constraints within the syntax of MCSs,
we obtained decidability and complexity results for the problem of whether a particular
MCS weakly or strongly satisfies a set of integrity constraints, and of repairing it in
the negative case.
We argued however that by keeping integrity constraints as an added layer on top of an MCS we
are able to separate intrinsic logical inconsistency from inconsistencies that may arise
e.g.~from improper changes to an individual context, which we want to detect and fix,
rather than propagate to other contexts.
Our examples show that we indeed capture the usual behaviour of integrity constraints in
several existing formalisms.

We also defined a notion of repair, consistent with the tradition in databases, and identified
new research problems related to which repairs should be preferred that arise in the MCS
scenario.
We intend to pursue this study further by developing a theory of active integrity
constraints, in the style of~\cite{Flesca2004}.

\paragraph{Acknowledgements.}
We would like to thank Gra\c ca Gaspar for introducing us to
the exciting topic of integrity constraints and for many fruitful
discussions.
We also thank the anonymous referees for many valuable suggestions
that improved the overall quality of this paper.
This work was supported by the Danish Council for Independent
Research, Natural Sciences, and by FCT/MCTES/PIDDAC under centre grant
to BioISI (Centre Reference: UID/MULTI/04046/2013).

%% file: paper.bbl
\begin{thebibliography}{10}

\bibitem{Abiteboul1988}
S.~Abiteboul.
\newblock Updates, a new frontier.
\newblock In M.~Gyssens, J.~Paredaens, and D.~van Gucht, editors, {\em
  ICDT'88}, volume 326 of {\em LNCS}, pages 1--18. Springer, 1988.

\bibitem{Abiteboul1995}
S.~Abiteboul, R.~Hull, and V.~Vianu.
\newblock {\em Foundations of Databases}.
\newblock Addison Wesley, 1995.

\bibitem{Chomicki1999}
M.~Arenas, L.E. Bertossi, and J.~Chomicki.
\newblock Consistent query answers in inconsistent databases.
\newblock In V.~Vianu and C.H. Papadimitriou, editors, {\em PODS 1999}, pages
  68--79. {ACM} Press, 1999.

\bibitem{Asirelli1985}
P.~Asirelli, M.~de Santis, and M.~Martelli.
\newblock Integrity constraints for logic databases.
\newblock {\em J. Log. Program.}, 2(3):221--232, 1985.

\bibitem{Brewka2007b}
G.~Brewka and T.~Eiter.
\newblock Equilibria in heterogeneous nonmonotonic multi-context systems.
\newblock In {\em AAAI2007}, pages 385--390. AAAI Press, 2007.

\bibitem{Brewka2011b}
G.~Brewka, T.~Eiter, and M.~Fink.
\newblock Nonmonotonic multi-context systems: {A} flexible approach for
  integrating heterogeneous knowledge sources.
\newblock In M.~Balduccini and T.~Cao Son, editors, {\em Logic Programming,
  Knowledge Representation, and Nonmonotonic Reasoning - Essays Dedicated to
  Michael Gelfond on the Occasion of His 65th Birthday}, volume 6565 of {\em
  LNCS}, pages 233--258. Springer, 2011.

\bibitem{Brewka2011}
G.~Brewka, T.~Eiter, M.~Fink, and A.~Weinzierl.
\newblock Managed multi-context systems.
\newblock In T.~Walsh, editor, {\em IJCAI 2011}, pages 786--791. IJCAI/AAAI,
  2011.

\bibitem{Caroprese2012}
L.~Caroprese, I.~Trubitsyna, M.~Truszczynski, and E.~Zumpano.
\newblock The view-update problem for indefinite databases.
\newblock In L.~Fari{\~n}as~del Cerro, A.~Herzig, and J.~Mengin, editors, {\em
  JELIA 2012}, volume 7519 of {\em LNCS}, pages 134--146. Springer, 2012.

\bibitem{Caroprese2007}
L.~Caroprese and E.~Zumpano.
\newblock Consistent data integration in {P2P} deductive databases.
\newblock In H.~Prade and V.S. Subrahmanian, editors, {\em SUM 2007}, volume
  4772 of {\em LNCS}, pages 230--243. Springer, 2007.

\bibitem{Caroprese2014}
L.~Caroprese and E.~Zumpano.
\newblock Dealing with incompleteness and inconsistency in {P2P} deductive
  databases.
\newblock In B.C. Desai, A.M. Almeida, J.~Bernardino, and E.~Ferreira~Gomes,
  editors, {\em IDEAS 2014}, pages 124--131. {ACM}, 2014.

\bibitem{LCF2014}
L.~Cruz{-}Filipe.
\newblock Optimizing computation of repairs from active integrity constraints.
\newblock In C.~Beierle and C.~Meghini, editors, {\em FoIKS 2014}, volume 8367
  of {\em LNCS}, pages 361--380. Springer, 2014.

\bibitem{KMIS2015}
L.~Cruz-Filipe, M.~Franz, A.~Hakhverdyan, M.~Ludovico, I.~Nunes, and
  P.~Schneider-Kamp.
\newblock {repAIrC}: A tool for ensuring data consistency by means of active
  integrity constraints.
\newblock In A.L.N. Fred, J.L.G. Dietz, D.~Aveiro, K.~Liu, and J.~Filipe,
  editors, {\em KMIS}, pages 17--26. SciTePress, 2015.

\bibitem{CGN2014}
L.~Cruz-Filipe, G.~Gaspar, and I.~Nunes.
\newblock Information flow within relational multi-context systems.
\newblock In K.~Janowicz, S.~Schlobach, P.~Lambrix, and E.~Hyv{\"{o}}nen,
  editors, {\em EKAW 2014}, volume 8876 of {\em LNAI}, pages 97--108. Springer,
  2014.

\bibitem{CHN2013b}
L.~Cruz-Filipe, R.~Henriques, and I.~Nunes.
\newblock Description logics, rules and multi-context systems.
\newblock In K.~McMillan, A.~Middeldorp, and A.~Voronkov, editors, {\em
  LPAR-19}, volume 8312 of {\em LNCS}, pages 243--257. Springer, December 2013.

\bibitem{DaoTran2010}
M.~Dao{-}Tran, T.~Eiter, M.~Fink, and T.~Krennwallner.
\newblock Distributed nonmonotonic multi-context systems.
\newblock In F.~Lin, U.~Sattler, and M.~Truszczynski, editors, {\em {KR} 2010}.
  {AAAI} Press, 2010.

\bibitem{Eiter2011}
T.~Eiter, M.~Fink, G.~Ianni, and P.~Sch{\"{u}}ller.
\newblock The {IMPL} policy language for managing inconsistency in
  multi-context systems.
\newblock In H.~Tompits, S.~Abreu, J.~Oetsch, J.~P{\"{u}}hrer, D.~Seipel,
  M.~Umeda, and A.~Wolf, editors, {\em {INAP}/{WLP} 2011}, volume 7773 of {\em
  LNCS}, pages 3--26. Springer, 2011.

\bibitem{Eiter2014}
T.~Eiter, M.~Fink, P.~Sch{\"{u}}ller, and A.~Weinzierl.
\newblock Finding explanations of inconsistency in multi-context systems.
\newblock {\em Artif. Intell.}, 216:233--274, 2014.

\bibitem{Eiter1992}
T.~Eiter and G.~Gottlob.
\newblock On the complexity of propositional knowledge base revision, updates,
  and counterfactuals.
\newblock {\em Artif. Intell.}, 57(2-3):227--270, 1992.

\bibitem{Fang2011}
M.~Fang, W.~Li, and R.~Sunderraman.
\newblock Maintaining integrity constraints among distributed ontologies.
\newblock In {\em {CISIS} 2011}, pages 184--191. {IEEE}, 2011.

\bibitem{Fink2011}
M.~Fink, L.~Ghionna, and A.~Weinzierl.
\newblock Relational information exchange and aggregation in multi-context
  systems.
\newblock In J.P. Delgrande and W.~Faber, editors, {\em LPNMR 2011}, volume
  6645 of {\em LNCS}, pages 120--133. Springer, 2011.

\bibitem{Flesca2004}
S.~Flesca, S.~Greco, and E.~Zumpano.
\newblock Active integrity constraints.
\newblock In E.~Moggi and D.~Scott~Warren, editors, {\em PPDP 2004}, pages
  98--107. ACM, 2004.

\bibitem{Gallaire1984}
H.~Gallaire, J.~Minker, and J.-M. Nicolas.
\newblock Logic and databases: {A} deductive approach.
\newblock {\em {ACM} Comput. Surv.}, 16(2):153--185, 1984.

\bibitem{Goncalves2014}
R.~Gon{\c{c}}alves, M.~Knorr, and J.~Leite.
\newblock Evolving multi-context systems.
\newblock In T.~Schaub, G.~Friedrich, and B.~O'Sullivan, editors, {\em {ECAI}
  2014}, volume 263 of {\em Frontiers in Artificial Intelligence and
  Applications}, pages 375--380. {IOS} Press, 2014.

\bibitem{Kakas1990}
A.C. Kakas and P.~Mancarella.
\newblock Database updates through abduction.
\newblock In D.~McLeod, R.~Sacks-Davis, and H.-J. Schek, editors, {\em VLDB
  1990}, pages 650--661. Morgan Kaufmann, 1990.

\bibitem{Mayol2003}
E.~Mayol and E.~Teniente.
\newblock Consistency preserving updates in deductive databases.
\newblock {\em Data Knowl. Eng.}, 47(1):61--103, 2003.

\bibitem{Motik2009}
B.~Motik, I.~Horrocks, and U.~Sattler.
\newblock Bridging the gap between {OWL} and relational databases.
\newblock {\em Web Semantics: Science, Services and Agents on the World Wide
  Web}, 7(2), 2011.

\bibitem{Pourebrahimi05asurvey}
B.~Pourebrahimi, K.~Bertels, and S.~Vassiliadis.
\newblock A survey of peer-to-peer networks.
\newblock In {\em ProRISC 2005}, 2005.

\bibitem{Eiter2010}
J.~P{\"{u}}hrer, S.~Heymans, and T.~Eiter.
\newblock Dealing with inconsistency when combining ontologies and rules using
  dl-programs.
\newblock In L.~Aroyo, G.~Antoniou, E.~Hyv{\"{o}}nen, A.~ten Teije,
  H.~Stuckenschmidt, L.~Cabral, and T.~Tudorache, editors, {\em {ESWC}(1)
  2010}, volume 6088 of {\em LNCS}, pages 183--197. Springer, 2010.

\bibitem{Schlipf1995}
J.S. Schlipf.
\newblock Complexity and undecidability results for logic programming.
\newblock {\em Annals of Mathematics and Artificial Intelligence},
  15(3--4):257--288, 1995.

\bibitem{Chomicki2010}
S.~Staworko and J.~Chomicki.
\newblock Consistent query answers in the presence of universal constraints.
\newblock {\em Inf. Syst.}, 35(1):1--22, 2010.

\bibitem{Tasharrofi2014}
S.~Tasharrofi and E.~Ternovska.
\newblock Generalized multi-context systems.
\newblock In C.~Baral, G.~de~Giacomo, and T.~Eiter, editors, {\em {KR} 2014}.
  {AAAI} Press, 2014.

\bibitem{Teniente1995}
E.~Teniente and A.~Oliv{\'e}.
\newblock Updating knowledge bases while maintaining their consistency.
\newblock {\em VLDB J.}, 4(2):193--241, 1995.

\bibitem{Ullman1988}
J.D. Ullman.
\newblock {\em Principles of Database and Knowledge-Base Systems, Volume {I}}.
\newblock Computer Science Press, 1988.

\bibitem{Vardi1982}
M.Y. Vardi.
\newblock The complexity of relational query languages (extended abstract).
\newblock In H.R. Lewis, B.B. Simons, W.A. Burkhard, and L.H. Landweber,
  editors, {\em STOC 1982}, pages 137--146. {ACM}, 1982.

\bibitem{Weinzierl2011}
A.~Weinzierl.
\newblock Advancing multi-context systems by inconsistency management.
\newblock In S.~Bragaglia, C.~Dam{\'{a}}sio, M.~Montali, A.D. Preece, C.J.
  Petrie, M.~Proctor, and U.~Straccia, editors, {\em RuleML2011@BRF Challenge},
  volume 799 of {\em {CEUR} Workshop Proceedings}. CEUR-WS.org, 2011.

\end{thebibliography}
